\documentclass{article}
\usepackage{amsmath,amssymb,amsthm}
\usepackage[noadjust]{cite}
\usepackage{comment}

\usepackage{tikz}
\usetikzlibrary{decorations.pathreplacing}
\usetikzlibrary{calc}

\newcommand{\Erdos}{Erd\H{o}s}
\newcommand{\Althofer}{Alth\"{o}fer}
\newcommand{\Szabo}{Szab\'{o}}

\newcommand{\Fp}{\mathbb{F}_p}
\newcommand{\Bk}{\mathcal{B}}
\newcommand{\Lift}{\mathcal{L}}
\newcommand{\dist}{\operatorname{dist}}
\newcommand{\edge}[2]{\{#1,#2\}}

\newcommand{\ex}{\texttt{ex}}
\newcommand{\linelength}{\texttt{line-length}}

\usepackage{todonotes}

\title{Unconditional Lower Bounds for Degree Fault Tolerant Spanners}
\date{}

\usepackage{geometry}
\usepackage{hyperref}
\newtheorem{theorem}{Theorem}
\theoremstyle{definition}
\newtheorem{definition}[theorem]{Definition}
\newtheorem{lemma}[theorem]{Lemma}

\author{Greg Bodwin and Aleksey Lopez\\University of Michigan EECS\\ \texttt{\{bodwin, lopezag\}@umich.edu}}

\begin{document}

\maketitle

\thispagestyle{empty}
\setcounter{page}{0}

\begin{abstract}
We study multiplicative graph spanners in the \emph{$f$-degree fault tolerant ($f$-DFT)} model, in which the spanner must approximately preserve distances even after any subset of edges of maximum degree $f$ temporarily ``fails'' and is removed from the graph.
We prove that there are $n$-node lower bound graphs for which any $f$-DFT $(2k-1)$-stretch spanner $H$ must have size
$$|E(H)| \ge \Omega\left( f^{1-1/k} n^{1+1/k}\right).$$
This matches a lower bound that was previously only known to hold conditionally, under the 1963 \emph{girth conjecture} of \Erdos{}.
It also matches the current upper bounds, up to a factor of $\texttt{exp}(k)$.
Our proof is an analysis of the so-called \emph{Wenger graphs} (J.\ Comb.\ Theory 1991), via their recent reinterpretation by \Szabo{} and by Conlon (Am.\ Math.\ Monthly 2021).
\end{abstract}

\clearpage

\section{Introduction}

\subsection{Spanners and the Girth Conjecture}

A \emph{multiplicative spanner} is a primitive in graph theory and algorithms, which sparsifies an input graph while approximately preserving its shortest path distances.
Spanners were abstracted in the late '80s after being implicitly studied in networking, distributed algorithms, etc (see survey \cite{ABSHJKS20}).

\begin{definition} [Spanners \cite{PS89}]
Given a graph $G = (V, E, w)$, a subgraph $H = (V, E', w)$ is a $k$-spanner of $G$ if for all nodes $s, t$, we have $\dist_H(s, t) \le k \cdot \dist_G(s, t)$.
\end{definition}

Spanners and their variants are often studied from the \emph{extremal} perspective, where the goal is to prove that one can always construct a spanner with a good tradeoff between size and error.
The following classic theorem of \Althofer{}, Das, Dobkin, Joseph, and Soares \cite{ADDJS93} conditionally settles the extremal tradeoff available for multiplicative spanners:

\begin{theorem} [\cite{ADDJS93}] \label{thm:introspan}
For all positive integers $n, k$, every $n$-node graph has a $(2k-1)$-spanner on $O(n^{1+1/k})$ edges.
Moreover, assuming the girth conjecture \cite{girth}, for all positive integers $k$ there are families of $n$-node graphs on which this size bound cannot be improved to $o(n^{1+1/k})$.
\end{theorem}

The \emph{girth conjecture} referenced in this theorem, attributed to \Erdos{} \cite{girth}, says the following.
Let us denote by $\ex(n, C_{\le 2k})$ the maximum possible number of edges in an $n$-node graph that does not have any cycles of length $\le 2k$ as subgraphs.
There is a simple folklore counting argument, called the \emph{Moore bound}, that implies
$$\ex(n, C_{\le 2k}) \le O\left( n^{1+1/k}\right).$$
The girth conjecture posits that this bound is asymptotically best possible for all $k$.
The girth conjecture is an old and very difficult problem in extremal graph theory, likely beyond the reach of current techniques.
It has been confirmed only for $k \in \{1, 2, 3, 5\}$ \cite{Wenger91, Tits59}, with no new cases settled since the 1950s.
In fact, the community has more recently started to regard it with some skepticism.
For example, in an article from 2021, Conlon \cite{Conlon21} writes:
\begin{center}
\emph{``It is now quite commonplace to believe that the true value of certain extremal numbers lie below what the classical arguments give. We suspect that this should already be the case for $C_8$, that is, that $\ex(n, C_8) = o(n^{5/4})$.''}
\end{center}
Here $\ex(n, C_8)$ is the maximum possible number of edges in an $n$-node graph with no $C_8$ subgraph.
It follows from the definitions that $\ex(n, C_8) \ge \ex(n, C_{\le 8})$, and so the suggested bound of $o(n^{5/4})$ would imply that the girth conjecture is false already for $k=4$.
Since this seems difficult to prove, Conlon goes on to give a more attainable conjecture, very roughly stating that the techniques from incidence geometry used to confirm the girth conjecture for $k \in \{1, 2, 3, 5\}$ cannot be extended to other values of $k$.
This conjecture is related to the classical Feit-Higman Theorem \cite{FH64}, which is sometimes interpreted as further evidence against the girth conjecture.\footnote{The constructions implying the girth conjecture for $k \in \{1, 2, 3, 5\}$ arise from highly symmetric incidence structures known as \emph{finite generalized polygons}.  The Feit-Higman Theorem constrains the existence of finite generalized polygons, and in particular, it implies that finite generalized polygons cannot be used to prove the girth conjecture for any further values of $k$.}
For further technical discussion and references on the girth conjecture and its surrounding literature, we recommend the excellent set of lecture notes by \Szabo{} \cite{Szabo18}.


\subsection{Fault Tolerant Spanners}

Spanners are often used in distributed systems and networking, where there is a possibility of nodes or edges in the input graph $G$ temporarily \emph{failing} and becoming unusable while they await repair.
A desirable feature in these application domains is for a spanner $H$ to retain its distance approximation guarantees even after any ``reasonable'' failure event occurs.

The first formal model of fault tolerant spanners was introduced by Levcopoulos, Narasimhan, and Smid \cite{LNS98}.
They defined $H$ to be an $f$-vertex (edge) fault tolerant $k$-spanner of $G$ if, for any set $F$ of $|F| \le f$ vertices (edges), the subgraph $H \setminus F$ is still a $k$-spanner of $G \setminus F$.
The initial results in this model were for geometric graphs, which has remained an area of intensive focus for the community \cite{LNS98,Lukovszki99,CZ04,NS07,BHO20,Solomon14,BDMS13,CLN15,CLNS15,LST23}.
Results for general graphs were obtained later by Chechik, Langberg, Peleg, and Roditty~\cite{CLPR10}, which has inspired a similarly long line of followup work~\cite{DK11podc,BDPV18,BP19,DR20,BDR21,BDR22,Parter22}.
By now, this model of fault tolerance is mostly well understood.
We have optimal existential bounds \cite{BDPV18, BP19} and near-optimal algorithms \cite{Parter22} for the setting of vertex faults.
We also have optimal existential bounds for edge faults for infinitely many values of spanner error $k$ \cite{BDR21, BDPV18}, although some knowledge gaps remain.
As usual, the lower bounds in both settings are conditioned on the girth conjecture.\footnote{In the setting of vertex faults (but not edge faults), we remark that the upper bounds and lower bounds both work \emph{by reduction} to $\ex(n, C_{\le 2k})$, and hence proving an unconditional tight lower bound would be equivalent to proving the girth conjecture itself \cite{BP19, BDPV18}.  In the setting of edge faults, for higher fault tolerance parameters $f$, the lower bounds can also be proved by conditioning on a certain bipartite extension of the girth conjecture; see \cite{BDR22}, Section 2.1 for discussion.}

More recently, the community has turned toward new and more expressive models of fault tolerance, that strengthen the original model without harming the spanner size too much.
One such model is \emph{color-fault tolerant spanners}, a compelling model of correlated failure introduced recently by Petruschka, Sapir, and Tzalik \cite{PST25} (see also \cite{PST25, PPST25, PST24b, PP25}).
Optimal or near-optimal bounds are known for this setting as well, again with the lower bounds conditioned on the girth conjecture.\footnote{Again, the upper and lower bounds in the color-fault model work by reduction to $\ex(n, C_{\le 2k})$, and so proving an unconditional lower bound is equivalent to proving the girth conjecture \cite{PST25}.}
Another more expressive fault tolerance model, which will be a focus of this work, is \emph{degree fault tolerant spanners}:

\begin{figure}[t]
\begin{center}
\begin{tikzpicture}[scale=1.2, every node/.style={font=\small}]
  \fill (0,1) circle (2pt);
  \fill (4,1) circle (2pt);

  \fill (2,2) circle (2pt);
  \fill (2,1) circle (2pt);
  \fill (2,0) circle (2pt);

  \draw (0,1) -- node [midway, sloped, red] {\Huge \bf $\times$} (2,2) -- (4,1);
  \draw (0,1) -- (2,1) -- (4,1);
  \draw (0,1) -- (2,0) -- node [midway, sloped, red] {\Huge \bf $\times$} (4,1);

  \draw (2,2) -- (2,1);
  \draw (2,1) -- (2,0);

\begin{scope}[shift={(6, 0)}]
  \fill (0,1) circle (2pt);
  \fill (4,1) circle (2pt);

  \fill (2,2) circle (2pt);
  \fill (2,1) circle (2pt);
  \fill (2,0) circle (2pt);

  \draw (0,1) -- node [midway, sloped, red] {\Huge \bf $\times$} (2,2);
  \draw (0,1) -- (2,1) -- (4,1);
  \draw (0,1) -- (2,0) -- node [midway, sloped, red] {\Huge \bf $\times$} (4,1);

  \draw (2,2) -- (2,1);
  \draw (2,1) -- (2,0);
  \end{scope}
\end{tikzpicture}
\end{center}
\caption{An input graph $G$ (left) and a $1$-DFT $3$-spanner $H$ (right).  Note that after any degree-1 fault set is removed from both $G$ and $H$ (such as the pair of edges marked with red X's), the remaining part of $H$ is still a $3$-spanner of the remaining part of $G$.}
\end{figure}
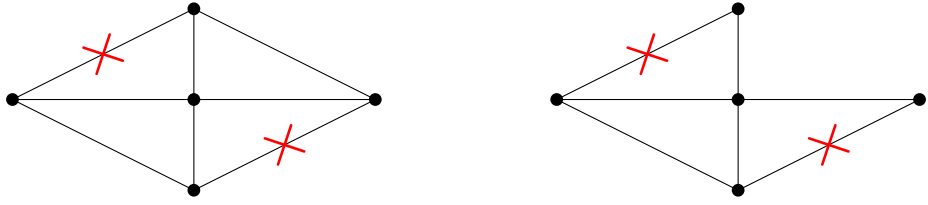

\begin{definition} [Degree Fault Tolerant Spanners \cite{BHP24}]
Given a graph $G = (V, E, w)$, a subgraph $H = (V, E', w)$ is an $f$-degree fault tolerant ($f$-DFT) $k$-spanner of $G$ if, for any edge subset $F \subseteq E$ such that $F$ contains at most $f$ edges incident on each node, we have
$$\dist_{H \setminus F}(s, t) \le k\cdot \dist_{G \setminus F}(s, t) \qquad \text{for all nodes } s, t.$$
\end{definition}

An $f$-DFT spanner can tolerate many more edge failures at once than in the original model of fault tolerance; for example, it can handle $O(nf)$ edge failures, so long as they are spread throughout the graph such that at most $f$ failing edges are incident to any particular node.
Despite this possibility of massive failure, Bodwin, Haeupler, and Parter \cite{BHP24} showed that fairly reasonable existential size bounds are available, not much worse than the ones known for the original model:
\begin{theorem} [\cite{BHP24}]
For all positive integers $k$, every $n$-node graph $G$ has an $f$-DFT $(2k-1)$-spanner $H$ on
$$|E(H)| \le O\left( \exp(k) \cdot f^{1-1/k} \cdot n^{1+1/k}\right)$$
edges.
Moreover, assuming the girth conjecture, there is a family of $n$-node input graphs $G$ for which any $f$-DFT $(2k-1)$-spanner $H$ must have
$$|E(H)| \ge \Omega\left( f^{1-1/k} \cdot n^{1+1/k}\right).$$
\end{theorem}

Thus conditionally nearly-tight bounds are known for this problem, up to a factor of $\exp(k)$, and assuming the girth conjecture.
Followup work by Parter and Tzalik \cite{PT25} obtained matching upper and lower bounds in the reachability ($k=\infty$) setting, and Biniaz, Carufel, Maheshwari, and Smid \cite{BDMS24} obtained results in the DFT model for metric and geometric spanners.

\subsection{Our Results}

The contribution of this paper is to show that the additional strength of the DFT model is enough to obtain \emph{unconditional} lower bounds.
We show:

\begin{theorem} \label{thm:intromaindft}
For all positive integers $k$, there are $n$-node input graphs $G$ for which any $f$-DFT $(2k-1)$-spanner $H$ must have
$$|E(H)| \ge \Omega\left( f^{1-1/k} \cdot n^{1+1/k}\right).$$
\end{theorem}

These lower bounds are both optimal up to $O_k(1)$ factors, and they hold via explicit constructions that do not assume the girth conjecture.
This makes DFT spanners essentially the first variants of multiplicative spanners where one can obtain unconditionally near-optimal lower bounds.
For essentially every other variant, the best unconditional lower bounds are inherited from the lower bounds for the function $\ex(n, C_{\le 2k})$, where the state-of-the-art is from a 1995 work of Lazebnik, Ustimenko, and Woldar \cite{LUW95}.
These bounds are polynomially below the bound posited by the girth conjecture, meaning that our results represent polynomial improvements in the known unconditional bounds.

These theorems may also have independent interest from the standpoint of extremal combinatorics.
In some sense, they state that the DFT model captures the gap between known algebraic techniques and the structure that would be needed to prove the girth conjecture, if possible.
More specifically, our proofs analyze the \emph{Wenger graphs} \cite{Wenger91}, which are an elegant algebraic construction that notably achieves optimal lower bounds on the function $\ex(n, C_{2k})$ for $k \in \{1, 2, 3, 5\}$.\footnote{The Wenger graphs also imply the girth conjecture for $k \in \{1, 2, 3\}$, but not for $k=5$: although they do not have $C_{10}$ subgraphs, they do have $C_8$ subgraphs.}
\Szabo{} \cite{Szabo18} and Conlon \cite{Conlon21} independently devised a reinterpretation of Wenger's original construction with a more algebraic flavor.
We adopt their interpretation, and make algebraic arguments to reason about the fine-grained structure of the short cycles that appear in these graphs.
The main part of our proofs can be viewed as a formalization of the sense in which the Wenger graphs come close to proving the girth conjecture; that is, their short cycles are sparsely ``spread out'' throughout the graph in a way that lets one destroy the relevant ones with a low-degree fault set.

\section{Proof of Main Result}

\subsection{Construction of Incidence System and Fault Sets}

We begin by constructing an incidence system, consisting of $n$ \emph{points} and $n$ \emph{lines} (subsets of the points).
This construction and its incidence graph are originally by Wenger \cite{Wenger91}, while our algebraic exposition follows \Szabo{} \cite{Szabo18} and Conlon \cite{Conlon21}.

\begin{itemize}
\item \textbf{(The parameters.)} The construction is parametrized by an arbitrary prime $p$ and an arbitrary integer $k \ge 1$.
We will have $n:=p^k$ points and lines in the construction.
We will use the standard shorthand $[p] := \{0, 1, \dots, p-1\}$.

\item \textbf{(The points.)} The set of points $A$ is the elements of the finite field $\Fp^k$.
We interpret $\Fp^k$ as a vector space with scalars in $\Fp$, and so we will perform scalar multiplication with scalars from $[p]$.

\item \textbf{(The lines.)} For any $t \in [p]$, let
$\sigma_t := (1,t,t^2,\dots,t^{k-1}) \in \Fp^k$;
we will refer to $\sigma_t$ as a \emph{slope vector}.
The set of lines $L$ is exactly the lines of slope $\sigma_t$ for any $t \in [p]$.
In other words, a subset of points forms a line $\ell$ iff there exists $t \in [p]$ such that $\ell$ is a maximal subset of points with the property that any two points in $\ell$ differ by a multiple of $\sigma_t$.
Formally:
$$L := \left\{ \big\{x + \lambda \sigma_t \ \mid \ \lambda \in [p] \big\} \mid x \in \Fp^k,  t \in [p] \right\}.$$

\item \textbf{(The graph.)} We define \(\Bk\) to be the bipartite \emph{incidence graph} between the set of points \(A\) and the set of lines \(L\).
That is: the left side is \(A\), the right side is \(L\), and there is an edge \(\edge{a}{\ell} \in E(\Bk)\) iff \(a\in \ell\).
\end{itemize}

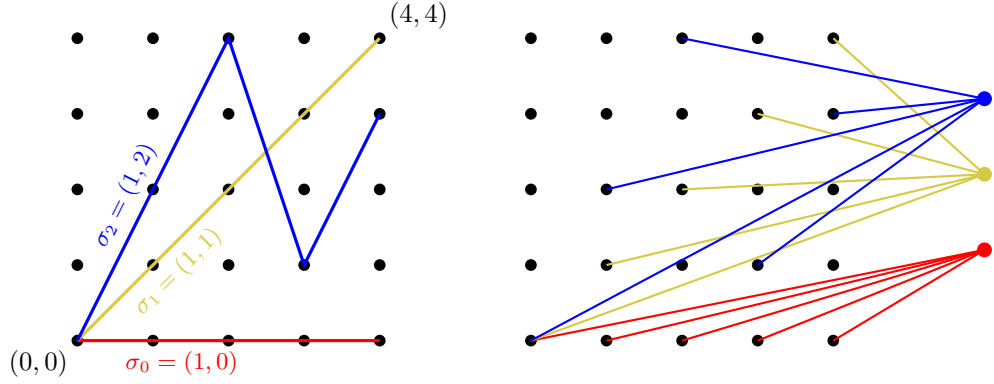
\begin{figure}[t]
\begin{center}
\begin{tikzpicture}[scale=1]
  \foreach \x in {0,1,2,3,4}{
    \foreach \y in {0,1,2,3,4}{
      \fill (\x,\y) circle (2.2pt);
    }
  }

  \node[below left] at (0,0) {$(0,0)$};
  \node[above right] at (4,4) {$(4,4)$};

  \draw[red, very thick]
    (0,0) to node [midway, below right, red] {\small $\sigma_0 = (1, 0)$} (1, 0) -- (4,0);
  
  \draw[yellow!80!black, very thick]
    (0,0) to node [midway, below right, yellow!80!black, sloped] {\small $\sigma_1 = (1, 1)$} (1,1) -- (2,2) -- (3,3) -- (4,4);

  \draw[blue, very thick]
    (0,0) to node [midway, above right, blue, sloped] {\small $\sigma_2 = (1, 2)$} (1,2) -- (2,4) -- (3,1) -- (4,3);

\begin{scope}[shift={(6, 0)}]
  \foreach \x in {0,1,2,3,4}{
    \foreach \y in {0,1,2,3,4}{
      \fill (\x,\y) circle (2.2pt);
    }
  }


  \fill[red] (6,1.2) circle (2.8pt);
  \fill[yellow!80!black] (6,2.2) circle (2.8pt);
  \fill[blue] (6,3.2) circle (2.8pt);


  \draw[red, thick] (6,1.2) -- (0,0);
  \draw[red, thick] (6,1.2) -- (1,0);
  \draw[red, thick] (6,1.2) -- (2,0);
  \draw[red, thick] (6,1.2) -- (3,0);
  \draw[red, thick] (6,1.2) -- (4,0);

  \draw[yellow!80!black, thick] (6,2.2) -- (0,0);
  \draw[yellow!80!black, thick] (6,2.2) -- (1,1);
  \draw[yellow!80!black, thick] (6,2.2) -- (2,2);
  \draw[yellow!80!black, thick] (6,2.2) -- (3,3);
  \draw[yellow!80!black, thick] (6,2.2) -- (4,4);

  \draw[blue, thick] (6,3.2) -- (0,0);
  \draw[blue, thick] (6,3.2) -- (1,2);
  \draw[blue, thick] (6,3.2) -- (2,4);
  \draw[blue, thick] (6,3.2) -- (3,1);
  \draw[blue, thick] (6,3.2) -- (4,3);
  \end{scope}
\end{tikzpicture}
\end{center}
\caption{(Left) The $25$ points and first $3$ lines used in the construction with $k=2$ and $p=5$.  (Right) The corresponding incidence graph.  In the full system, there would be $25$ lines and therefore $25$ nodes on the right side of the incidence graph.}
\end{figure}

The graph $\Bk$ has the following statistics:

\begin{lemma} [Size of $\Bk$] \label{lem:bksize}
The graph $\Bk$ has $2n$ nodes and $n^{1+1/k}$ edges.
\end{lemma}
\begin{proof}
There are $n=p^k$ nodes in $\Fp^k$.
The lines are indexed by $p^k$ possible choices of $x$ and $p$ possible choices of $t$, but each line can be generated in $p$ ways (by choosing any of its $p$ points for $x$), so there are $p^k \cdot p / p = p^k = n$ lines.
Thus $\Bk$ has $n$ nodes on each side of its bipartition, for $2n$ nodes in total.
Each line contains exactly $p$ points, so it generates $p$ edges in $\Bk$.
So there are $p^{k+1} = n^{(k+1)/k} = n^{1+1/k}$ edges in $\Bk$.
\end{proof}

Our goal will be to argue that \(\Bk\) is the only \(1\)-DFT \((2k-1)\)-spanner of itself.
This gives our desired lower bound in the special case $f=1$.
Later, in Section \ref{sec:cloudblowup}, we will describe a standard technique used to boost the lower bound to arbitrary $f \ge 1$.
The first step towards this goal is, for every edge $e$, to define a corresponding edge set $F_e$ of maximum degree $1$ (i.e., a matching) such that $e$ does not participate in any cycles of length $\le 2k$ in the graph $\Bk \setminus F_e$.
The next step in our proof is to define these sets $F_e$.
The following definitions will be useful:

\begin{definition} [Line and Path Properties]
Two lines $\ell, \ell' \in L$ are \textbf{parallel}, written $\ell \parallel \ell'$, if they are distinct and they have the same direction vector \(\sigma_t\).
A path $\pi$ in $\Bk$ with vertex sequence
$$
\pi=(a=a_0,\ell_0,a_1,\ell_1,\dots,\ell_{r-1},a_r=b)
$$
is called \textbf{\(\ell\)-transverse} if it is simple\footnote{Recall: a simple path is one that does not repeat nodes.} and none of the lines $\{\ell_0,\ell_1,\dots,\ell_{r-1}\}$ are equal to $\ell$ or parallel to \(\ell\).
Its \textbf{line-length} is the number of lines in $L$ corresponding to nodes in the path; for example, the above path $\pi$ has $\linelength(\pi) = r$.
\end{definition}

\begin{definition} [Failure Set $F_e$] \label{def:fe}
For any edge $e = \{a, \ell\} \in E(\Bk)$, we define \(F_e\subseteq E(\Bk)\setminus\{e\}\) to be the set of all edges \(\edge{b}{\ell'}\) such that \(\ell \parallel \ell'\), and there exists an \(\ell\)-transverse path from \(a\) to \(b\) of line-length $\le \frac{k}{2}-1$.\footnote{The expression $\frac{k}{2}-1$ is not necessarily integral, but line-length is necessarily integral, so this is equivalent to an upper bound on line-length of $\le \lfloor \frac{k}{2}-1\rfloor$.  Since the floor symbols are not technically necessary, we will omit them here and for other possibly-non-integral expressions encountered as upper bounds in the rest of this paper.}
\end{definition}

\begin{figure}
\begin{center}
\begin{tikzpicture}[scale=1.1, every node/.style={font=\small}]
  \draw[blue, thick] (0,-0.5) -- (0,2.2);
  \draw[blue, thick] (3.4,-0.5) -- (3.4,2.2);

  \node[blue, left] at (0,2.0) {$\ell$};
  \node[blue, right] at (3.4,2.0) {$\ell'$};

  \fill (0,0.3) circle (2.2pt);
  \node[left] at (-0.1,0.3) {$a$};

  \fill (3.4,1.0) circle (2.2pt);
  \node[right] at (3.4,1.1) {$b$};

  \fill (1.2,1.1) circle (1.7pt);
  \fill (2.2,0.2) circle (1.7pt);

  \draw[orange, thick] (0,0.3) -- (1.2,1.1);
  \draw[yellow!80!black, thick] (1.2,1.1) -- (2.2,0.2);
  \draw[green!60!black, thick] (2.2,0.2) -- (3.4,1.0);

  \draw[red, thick] (0,0.3) circle (0.1);
  \draw[red, ultra thick] (0,0.3-0.1) -- (0,0.3-0.6);
  \draw[red, ultra thick] (0,0.3+0.1) node [above left, red] {$e = \{a, \ell\}$} -- (0,0.3+0.6);

  \draw[red, thick] (3.4,1.0) circle (0.1);
  \draw[red, ultra thick] (3.4,1.0-0.1) -- (3.4,1.0-0.6);
  \draw[red, ultra thick] (3.4,1.0+0.1) node [above right, red] {$\{b, \ell'\} \in F_e$} -- (3.4,1.0+0.6);
\end{tikzpicture}
\end{center}
\caption{When considering the edge $e$ corresponding to the point-line incidence $a \in \ell$, we add to its associated fault set $F_e$ the edges corresponding to point-line incidences $b \in \ell'$, where $\ell \parallel \ell'$ and $b$ can be reached from $a$ by an $\ell$-transverse path of small line-length.}
\end{figure}
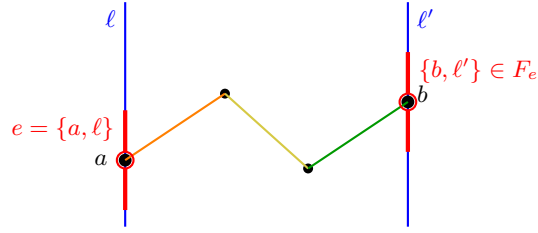

\subsection{Bound on Degree of $F_e$}

Our next goal is to prove that, for each edge $e \in E(\Bk)$, the associated fault set $F_e$ has maximum degree $1$.
We will need two supporting structural lemmas before the main proof.


\begin{lemma}
\label{lem:no-unique-slope}
For any closed, non-backtracking\footnote{Recall: a walk is \emph{closed} if it starts and ends at the same node, and it is \emph{non-backtracking} if it never uses the same edge twice in a row with opposite orientations.} walk $W$ in \(\Bk\) of length $|W| \le 2k$, there is no $t \in [p]$ such that \(W\) contains exactly one line of slope \(\sigma_t\).
(This allows the case where $W$ contains two instances of the same line $\ell$.)
\end{lemma}
\begin{proof}
Let the vertex sequence of $W$ be
$$
W=(a_0,\ell_0,a_1,\ell_1,\dots,\ell_{r-1},a_r=a_0),
$$
where each $a_j$ is a point and each $\ell_j$ is a line with slope vector $\sigma_j$, and $r \le k$.
Since \(a_j,a_{j+1}\in \ell_j\), there exists a scalar \(\lambda_j\in \Fp\) such that $a_{j+1}=a_j+\lambda_j \sigma_j$.
Moreover, since the walk is non-backtracking, we have $a_j \ne a_{j+1}$ and so these scalars are all nonzero.
Summing these equations around the walk $W$ gives the identity
$$
\sum_{j=0}^{r-1}\lambda_j \sigma_j = 0.
$$
However, it follows from the usual analysis of the Vandermonde matrix \cite{Conlon21} that any $k$ slope vectors $\{\sigma_j\}$ are linearly independent.
Thus, this identity is possible only if after grouping equal slope vectors together, their associated coefficients must sum to $0$.
Since each individual coefficient $\lambda_j$ is nonzero, this implies that each slope vector must have at least two associated coefficients, implying the lemma.
\end{proof}

\begin{lemma}[Unique endpoints]
\label{lem:canonical-endpoint}
Let \(\ell,\ell'\) be parallel lines in $L$ and let \(a\in \ell\).
Then there is a point $b \in \ell'$ such that every $\ell$-transverse path of line-length \(\le \frac{k}{2}-1\) that starts at $a$ and ends at a point in $\ell'$ specifically ends at $b$.
\end{lemma}
\begin{proof}
Suppose there are $\ell$-transverse paths with node sequences
$$
\pi_1=(a,\ell_0,a_1,\ell_1,\dots,\ell_{r-1},b) \quad \text{and} \quad \pi_2=(a,m_0,b_1,m_1,\dots,m_{s-1},b')
$$
ending at points \(b,b'\in \ell'\), of respective line-lengths $r, s$ both $\le \frac{k}{2}-1$.
Suppose for contradiction that \(b\neq b'\).
Concatenating \(\pi_1\) and the reverse of \(\pi_2\) with the node $\ell'$ placed between them produces a closed walk \(W\) in $\Bk$. Consider a minimal nonempty closed subwalk \(W' \subseteq W\) that contains an edge of \(\ell'\), which is a simple cycle (and hence non-backtracking).
We have
$$
\linelength(W') \le r+1+s \le 2 \cdot \left(\frac{k}{2} - 1\right) + 1 \le k-1.
$$
Hence the total length of $W'$ in $\Bk$ is $|W'| \le 2k$.
Moreover, the only line in \(W'\) parallel to \(\ell\) is \(\ell'\), since both \(\pi_1\) and \(\pi_2\) are \(\ell\)-transverse.
This contradicts Lemma~\ref{lem:no-unique-slope}.
Thus we have \(b=b'\), completing the proof.
\end{proof}

\begin{lemma}
\label{lem:Fe-degree1}
For any edge $e$, the edge set \(F_e \cup \{e\}\) has maximum degree \(1\).
\end{lemma}
\begin{proof}
Let $e =: \{a, \ell\}$.
We analyze $\deg_{F_e}(v)$ in cases:
\begin{itemize}
\item Suppose that $v$ corresponds to a point (rather than a line) in the incidence system.
Since $v$ lies on exactly one line $\ell'$ parallel to $\ell$ (or $v$ lies on $\ell$ itself), the only edge containing $v$ that might be in $F_e$ is $(v, \ell')$, so $\deg_{F_e}(v) \le 1$.

\item Suppose that $v$ corresponds to a line in the incidence system that is \emph{not} parallel to $\ell$.
Then by construction, $\deg_{F_e}(v) = 0$.

\item Suppose that $v$ corresponds to a line in the incidence system that is parallel to $\ell$.
By Lemma~\ref{lem:canonical-endpoint}, there is at most one point \(b \in v\) that is the endpoint of an \(\ell\)-transverse path from \(a\) of line-length $\le \frac{k}{2}-1$.
Hence \(v\) is incident to at most one edge of \(F_e\).

\item Suppose that $v = \ell$.
By construction, the only edge in $F_e \cup \{e\}$ incident to $v$ is $e$ itself.\qedhere
\end{itemize}
\end{proof}

\subsection{Correctness of Fault Set \(F_e\)}
\label{sec:blocking}

Our next goal is to prove that each edge $e$ does not participate in any cycles of length $\le 2k$ in the graph $\Bk \setminus F_e$.

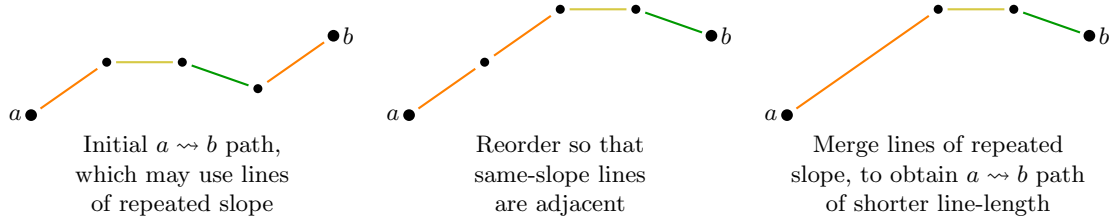
\begin{figure}[t]
\begin{center}
\begin{tikzpicture}[scale=1, every node/.style={font=\small}]
  \node (a) at (0,0.3) {};
  \fill (a) circle (2.2pt);
  \node[left] at (a) {$a$};

  \node (b) at (4.0,1.35) {};
  \fill (b) circle (2.2pt);
  \node[right] at (b) {$b$};

  \node (p1) at (1.0,1.0) {};
  \fill (p1) circle (1.7pt);

  \node (p2) at (2.0,1.0) {};
  \fill (p2) circle (1.7pt);

  \node (p3) at (3.0,0.65) {};
  \fill (p3) circle (1.7pt);

  \draw[orange, thick] (a) -- (p1);
  \draw[yellow!80!black, thick] (p1) -- (p2);
  \draw[green!60!black, thick] (p2) -- (p3);
  \draw[orange, thick] (p3) -- (b);

  \node [align=center] at (2, -0.5) {Initial $a \leadsto b$ path,\\ which may use lines\\of repeated slope};

  \begin{scope}[shift={(5, 0)}]
  \node (a) at (0,0.3) {};
  \fill (a) circle (2.2pt);
  \node[left] at (a) {$a$};

  \node (b) at (4.0,1.35) {};
  \fill (b) circle (2.2pt);
  \node[right] at (b) {$b$};

  \node (p1) at (1.0,1.0) {};
  \fill (p1) circle (1.7pt);

  \node (p2) at (2.0,1.7) {}; 
  \fill (p2) circle (1.7pt);

  \node (p3) at (3.0,1.7) {}; 
  \fill (p3) circle (1.7pt);

  \draw[orange, thick] (a) -- (p1);
  \draw[orange, thick] (p1) -- (p2);
  \draw[yellow!80!black, thick] (p2) -- (p3);
  \draw[green!60!black, thick] (p3) -- (b);

  \node [align=center] at (2, -0.5) {Reorder so that\\same-slope lines\\are adjacent};
  \end{scope}
  \begin{scope}[shift={(10, 0)}]
  
  \node (a) at (0,0.3) {};
  \fill (a) circle (2.2pt);
  \node[left] at (a) {$a$};

  \node (b) at (4.0,1.35) {};
  \fill (b) circle (2.2pt);
  \node[right] at (b) {$b$};


  \node (p2) at (2.0,1.7) {}; 
  \fill (p2) circle (1.7pt);

  \node (p3) at (3.0,1.7) {}; 
  \fill (p3) circle (1.7pt);

  \draw[orange, thick] (a) -- (p2);
  \draw[yellow!80!black, thick] (p2) -- (p3);
  \draw[green!60!black, thick] (p3) -- (b);

  \node [align=center] at (2, -0.5) {Merge lines of repeated\\slope, to obtain $a \leadsto b$ path\\of shorter line-length};
  \end{scope}
\end{tikzpicture}
\end{center}
\caption{A depiction of the ``path-straightening'' technique used to prove Lemma \ref{lem:path-straightening}.}
\end{figure}

\begin{lemma}[Path straightening]
\label{lem:path-straightening}
Let $\pi$ be an $\ell$-transverse $a \leadsto b$ path in $\Bk$ that uses lines with $z$ distinct slope vectors.
Then there is an $\ell$-transverse path $\pi'$ with the same endpoints $a \leadsto b$ that has line-length $\le z$.
\end{lemma}
\begin{proof}
Let $\pi'$ be an $a \leadsto b$ path in $\Bk$ of minimum length, among the paths with the property that every slope vector used by a line in $\pi'$ is also used by at least one line in $\pi$.
Note that, since $\pi$ is $\ell$-transverse, this implies that $\pi'$ is $\ell$-transverse as well.
Let the vertex sequence of $\pi'$ be
$$
\pi'=(a=a_0,\ell_0,a_1,\dots,\ell_{s-1},a_s=b).
$$
If all lines $\{\ell_i\}$ used by $\pi'$ have pairwise distinct slope vectors, then since $\pi'$ has at most $z$ distinct slope vectors, we get $s \le z$.
So $\pi'$ has line-length $\le z$ and it satisfies the lemma.

So, for the remaining part of the proof, assume towards contradiction that not all lines $\{\ell_i\}$ have distinct slope vectors.
For each index $i$, let $\Delta_i := a_{i+1} - a_i$.
We may therefore describe the sequence of point-nodes along $\pi'$ as
$$\left(a, a + \Delta_0, a + \Delta_0 + \Delta_1, \dots, a + \sum \limits_{i=0}^{s-1} \Delta_i = b\right).$$
Next, let $\phi$ be an arbitrary permutation on $[s]$.
We define $\pi_{\phi}$ as the walk in $\Bk$ whose sequence of point-nodes is
$$\left(a, a + \Delta_{\phi(0)}, a + \Delta_{\phi(0)} + \Delta_{\phi(1)}, \dots, a + \sum \limits_{i=0}^{s-1} \Delta_{\phi(i)} = b\right).$$
Note that adjacent point-nodes in this sequence differ by some vector $\Delta_i$, which is the slope of the line $\ell_i$, and so by construction there is indeed a line that contains these adjacent points.
This confirms that there is a unique walk $\pi_{\phi}$ with this sequence of point-nodes, so $\pi_{\phi}$ is well defined.

We assumed that a slope vector is used in $\pi'$ at least twice, so by choice of $\phi$, we may have
$$\Delta_{\phi(0)} = \lambda \cdot \Delta_{\phi(1)}$$
for some scalar $\lambda \in [p]$.
Thus, we may delete the second point-node along $\pi_{\phi}$, giving a shorter sequence
$$\left(a, a + \Delta_{\phi(0)} + \Delta_{\phi(1)}, \dots, a + \sum \limits_{i=0}^{s-1} \Delta_{\phi(i)} = b\right).$$
The difference between the first two point nodes is $(\lambda + 1) \cdot \Delta_{\phi(1)}$, so these points are still connected by a line of a viable slope vector. If $\lambda = -1$ then we simply omit that step, which still results in a shorter walk.
Thus there is a corresponding $a \leadsto b$ walk of length shorter than $\pi'$.
Deleting closed subwalks yields a path that contradicts minimality of the length of $\pi'$, completing the proof.
\end{proof}

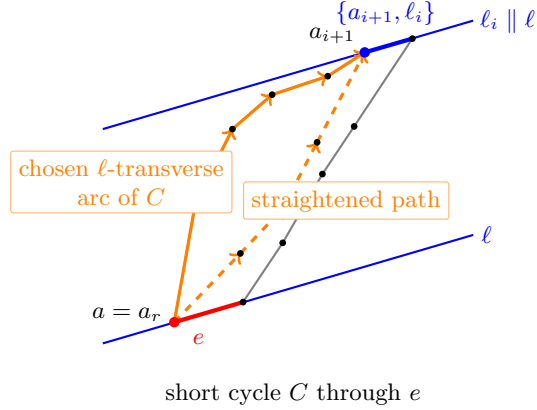
\begin{figure}[t]
\begin{center}
\begin{tikzpicture}[scale=0.7, every node/.style={font=\small}]

  \draw[blue, thick] (0.4,1.0) -- (7.4,3.05);
  \draw[blue, thick] (0.4,5.05) -- (7.4,7.10);

  \node[blue] at (7.65,3.05) {$\ell$};
  \node[blue] at (8.05,7.10) {$\ell_i \parallel \ell$};

  \coordinate (a) at (1.75,1.40);
  \coordinate (u) at (3.05,1.78);

  \coordinate (ai1) at (5.35,6.50);
  \coordinate (v) at (6.25,6.76);

  \coordinate (x1) at (4.65,6.05);
  \coordinate (x2) at (3.60,5.70);
  \coordinate (x3) at (2.85,5.05);
  \coordinate (x4) at (2.15,3.65);

  \coordinate (y1) at (3.80,2.90);
  \coordinate (y2) at (4.55,4.20);
  \coordinate (y3) at (5.15,5.10);

  \draw[gray, thick]
    (a) -- (u) -- (y1) -- (y2) -- (y3) -- (v) -- (ai1)
    -- (x1) -- (x2) -- (x3) -- (x4) -- cycle;

  \draw[orange, very thick, ->] (a) -- (x4);
  \draw[orange, very thick, ->] (x4) -- (x3);
  \draw[orange, very thick, ->] (x3) -- (x2);
  \draw[orange, very thick, ->] (x2) -- (x1);
  \draw[orange, very thick, ->] (x1) -- (ai1);

  \coordinate (s1) at (3.00,2.70);
  \coordinate (s2) at (3.80,3.55);
  \coordinate (s3) at (4.45,4.80);

  \draw[orange, very thick, dashed, ->] (a) -- (s1);
  \draw[orange, very thick, dashed, ->] (s1) -- (s2);
  \draw[orange, very thick, dashed, ->] (s2) -- (s3);
  \draw[orange, very thick, dashed, ->] (s3) -- (ai1);

  \draw[red, ultra thick] (a) -- (u);
  \node[red, fill=white, inner sep=1pt] at ($(a)!0.62!(u)+(-0.35,-0.55)$) {$e$};

  \draw[blue, ultra thick] (ai1) -- (v);
  \node[blue, fill=white, inner sep=1pt] at ($(ai1)!0.55!(v)+(-0.1,0.6)$)
    {$\{a_{i+1},\ell_i\}$};

  \fill[red] (a) circle (2.8pt);
  \node[left] at ($(a)+(-0.05,0.15)$) {$a=a_r$};

  \fill[blue] (ai1) circle (2.8pt);
  \node[above left] at (ai1) {$a_{i+1}$};

  \foreach \p in {u,v,x1,x2,x3,x4,y1,y2,y3,s1,s2,s3}{
    \fill[black] (\p) circle (1.8pt);
  }

  \node[orange, fill=white, draw=orange!60, rounded corners=1pt,
        inner sep=3pt, align=center]
    at (0.75,4.05)
    {chosen $\ell$-transverse\\ arc of $C$};

  \node[orange, fill=white, draw=orange!60, rounded corners=1pt,
        inner sep=3pt]
    at (5.00,3.70)
    {straightened path};

  \node[align=center] at (4.00,0.05)
    {short cycle \(C\) through \(e\)};

\end{tikzpicture}
\end{center}
\caption{Geometric view of the cycle-hitting argument. The solid orange path is the chosen \(\ell\)-transverse arc of the cycle, and the dashed orange path is the straightened path with the same endpoints.}
\end{figure}

\begin{lemma}
\label{lem:cycle-hit}
Every cycle $C$ in \(\Bk\) of length at most \(2k\) that contains an edge \(e=\edge{a}{\ell}\) also contains an edge of \(F_e\).
\end{lemma}
\begin{proof}
Let the vertex sequence of $C$ be
$$
C=(a,\ell,a_1,\ell_1,a_2,\ell_2,\dots,a_{r-1},\ell_{r-1},a_r=a)
$$
with $r \le k$.
By Lemma~\ref{lem:no-unique-slope}, at least one of the other lines in $C$ is parallel to \(\ell\).
Let $\ell_i$ be such a line, with the index $i$ selected as large as possible.
Our goal will now be to argue that $\edge{a_{i+1}}{\ell_i}\in F_e$.

There are $r$ distinct lines in $C$.
By Lemma \ref{lem:no-unique-slope}, each slope vector that is used by any of these $r$ lines is used by at least two of them, so there are at most $\frac{r}{2}$ distinct slope vectors used by lines in $C$.
Now consider the subpath $\pi \subseteq C$ with vertex sequence
$$
\pi=(a_r,\ell_{r-1},a_{r-1},\dots,\ell_{i+1},a_{i+1}).
$$
By maximal choice of $i$, this is an $\ell$-transverse path, and hence does not contain any lines whose slope vector matches $\ell$.
So there are at most $\frac{r}{2}-1$ distinct slope vectors used by the lines in $\pi$.
Thus, by Lemma \ref{lem:path-straightening}, there is an $\ell$-transverse path $\pi'$ with the same endpoints $(a_r, a_{i+1})$ that has line-length at most $\frac{r}{2}-1 \le \frac{k}{2}-1$.
Thus, by construction, the set $F_e$ associated to $e=\{a=a_r, \ell\}$ will include the edge $(a_{i+1}, \ell_i) \in C$.
\end{proof}

\subsection{Proof Wrapup \label{sec:cloudblowup}}

The following theorem certifies our lower bound for the special case $f=1$:

\begin{theorem} \label{thm:onedft}
\label{thm:main}
The only \(1\)-DFT \((2k-1)\)-spanner of \(\Bk\) is \(\Bk\) itself.
\end{theorem}

\begin{proof}
Let \(H\) be a proper subgraph of \(\Bk\), and we will show that \(H\) is not a \(1\)-DFT \((2k-1)\)-spanner of \(\Bk\).
Consider an edge
$$
e=\edge{u}{v}\in E(\Bk)\setminus E(H).
$$

Let \(F_e\) denote the set associated to this edge (as in Definition \ref{def:fe}).
By Lemma~\ref{lem:Fe-degree1}, the set \(F_e\) has maximum degree \(1\).
Next, let $\pi$ be a shortest $u \leadsto v$ path in the graph $H \setminus F_e$.
Since $e \notin E(H)$, we have $e \notin \pi$, and so $\pi \cup \{e\}$ is a cycle in $\Bk$.
This cycle contains $e$ but it does not contain any edge of $F_e$.
By Lemma~\ref{lem:cycle-hit}, this means $\pi \cup \{e\}$ contains $>2k$ edges.
Thus $\pi$ contains $>2k-1$ edges, so we have
$$
|E(\pi)| = \dist_{H\setminus F_e}(u,v) > 2k-1.
$$
On the other hand, \(e\notin F_e\), so we have
$
\dist_{\Bk\setminus F_e}(u,v)=1.
$
This implies that \(H\) is not a \(1\)-DFT \((2k-1)\)-spanner of \(\Bk\).
\end{proof}

We now upgrade the previous lower bound to an arbitrary degree parameter \(f\) using a standard cloud-blowup technique, which has appeared in many previous papers on fault tolerance \cite{PST25, BHP24, BDPV18}.

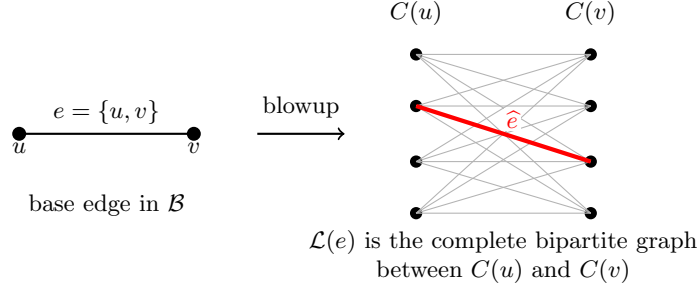
\begin{figure}[t]
\begin{center}
\begin{tikzpicture}[scale=1.05, every node/.style={font=\small}]

  \coordinate (u) at (0,0);
  \coordinate (v) at (2.2,0);

  \fill (u) circle (2.5pt);
  \fill (v) circle (2.5pt);

  \node[below] at (u) {$u$};
  \node[below] at (v) {$v$};

  \draw[thick] (u) -- node[above] {$e=\{u,v\}$} (v);

  \node[align=center] at (1.1,-0.9)
    {base edge in \(\Bk\)};

  \draw[->, thick] (3.0,0) -- (4.1,0);
  \node[above] at (3.55,0.08) {blowup};

  \foreach \i/\y in {1/1.0,2/0.35,3/-0.35,4/-1.0}{
    \coordinate (u\i) at (5.0,\y);
    \coordinate (v\i) at (7.2,\y);
    \fill (u\i) circle (2.2pt);
    \fill (v\i) circle (2.2pt);
  }

  \node[above] at (5.0,1.25) {$C(u)$};
  \node[above] at (7.2,1.25) {$C(v)$};

  \foreach \i in {1,2,3,4}{
    \foreach \j in {1,2,3,4}{
      \draw[gray!60, thin] (u\i) -- (v\j);
    }
  }

  \draw[red, ultra thick] (u2) -- (v3);
  \node[red, fill=white, inner sep=1pt]
    at ($(u2)!0.55!(v3)+(0,0.22)$)
    {$\widehat e$};

  \node[align=center] at (6.1,-1.55)
    {\(\Lift(e)\) is the complete bipartite graph\\
    between \(C(u)\) and \(C(v)\)};

\end{tikzpicture}
\end{center}
\caption{A lifted edge in the cloud blowup. Each base vertex is replaced by a cloud of \(f\) copies, and each base edge \(e=\{u,v\}\) is replaced by the complete bipartite graph \(\Lift(e)\) between \(C(u)\) and \(C(v)\).}
\end{figure}

\begin{definition}[Cloud blowup]
For an integer \(f \ge 1\), the $f$-cloud blowup of the base graph is
$$\Bk^{\langle f \rangle} := \Bk \times [f],$$
where $[f]$ denotes the edgeless graph on $f$ vertices, and $\times$ denotes the $f$-cloud blowup operation defined below.
In other words, it is the graph obtained from $\Bk$ by the following steps:
\begin{itemize}
    \item For each vertex \(x\in V(\Bk)\), create a ``cloud'' of $f$ nodes, $C(x):=\{(x,i): i\in [f]\}$.
    
    \item For each base edge \(\edge{x}{y}\in E(\Bk)\), lift it to a complete bipartite graph between the clouds \(C(x)\) and \(C(y)\).
    That is,
    \[
    \Lift(\edge{x}{y}):=\left\{\edge{(x,i)}{(y,j)} \mid i,j\in [f]\right\} \quad \text{and} \quad E\big(\Bk^{\langle f\rangle}\big)
    := \bigcup \limits_{\edge{x}{y} \in E(\Bk)} \Lift(\edge{x}{y}).
    \]
\end{itemize}
\end{definition}

The following lemma computes how the cloud blowup affects the parameters of the graph:

\begin{lemma}[Size of Blown-Up Graph]
Letting $2n$ be the number of vertices in $\Bk$, the corresponding blown-up graph $\Bk^{\langle f \rangle}$ has $N := 2nf$ vertices and $\Theta(f^{1-1/k}N^{1+1/k})$ edges.
\end{lemma}
\begin{proof}
Each node in $\Bk$ corresponds to \(f\) copies in $\Bk^{\langle f \rangle}$, and each edge in $\Bk$ corresponds to $f^2$ edges in $\Bk^{\langle f \rangle}$.
Thus we have
\begin{align*}
\left|E\left(\Bk^{\langle f \rangle}\right)\right| &= f^2 \cdot \left|E\left(\Bk\right)\right|\\
&= \Theta\left( f^2 n^{1+1/k} \right) \tag*{Lemma \ref{lem:bksize}}\\
&= \Theta\left( f^{1-1/k} (nf)^{1+1/k}\right)\\
&= \Theta\left( f^{1-1/k} N^{1+1/k}\right). \tag*{\qedhere}
\end{align*}
\end{proof}


From the previous lemma, it suffices to prove that $\Bk^{\langle f \rangle}$ is the only $f$-DFT $(2k-1)$-spanner of itself.
To do so, let us first define lifted fault sets:
\begin{definition} [Lifted Fault Sets $\widehat{F}_e$]
For an edge $e \in E(\Bk)$ and an edge $\widehat{e} \in \Lift(e)$, we define $\widehat{F}_{\widehat{e}}$ to be the set of all lifted edges from $F_e \cup \{e\}$, except for $\widehat{e}$ itself.
That is:
$$
\widehat F_{\widehat e}
:=
\left(\bigcup_{g \in F_e} \Lift(g)\right)
\cup
\big(\Lift(e) \setminus \{\widehat e\}\big).
$$
\end{definition}

The next lemma shows that the lifted fault sets have degree $\le f$:

\begin{lemma}[Lifted fault set has maximum degree \(f\)]
For any edge $\widehat{e}$, the corresponding edge set \(\widehat F_{\widehat e}\) has maximum degree \(\le f\) in \(\Bk^{\langle f\rangle}\).
\end{lemma}
\begin{proof}
Since \(F_e \cup \{e\}\) has maximum degree \(1\) in the base graph, each base vertex \(x\in V(\Bk)\) is incident to at most one edge of \(F_e \cup \{e\}\).
Hence each lifted vertex in $\Bk^{\langle f \rangle}$ is incident to at most one lifted biclique, and within that biclique it is incident to at most \(f\) edges. \qedhere

\end{proof}

We are finally ready to wrap up the proof:

\begin{theorem}[Cloud-blowup lower bound]
\label{thm:cloud-blowup}
For any positive integer $f$, the only \(f\)-DFT \((2k-1)\)-spanner of \(\Bk^{\langle f \rangle}\) is \(\Bk^{\langle f \rangle}\) itself.
\end{theorem}
\begin{proof}
Since the following cloud-lift is standard in the previous literature, we will sketch parts of the proof.
Let \(H\) be a proper subgraph of \(\Bk^{\langle f \rangle}\). Choose a missing lifted edge
$$
\widehat e=\edge{(u,i)}{(v,j)} \in E\left(\Bk^{\langle f \rangle}\right) \setminus E(H),
$$
and let \(e=\edge{u}{v}\in E(\Bk)\) be the corresponding base edge; that is, $\widehat{e} \in \Lift(e)$.
Consider the associated fault set \(\widehat F_{\widehat e}\).
By the previous lemma, \(\widehat F_{\widehat e}\) has maximum degree at most \(f\).
Next, let $\pi$ be a shortest $(u, i) \leadsto (v, j)$ path in the graph $H \setminus \widehat{F}_{\widehat{e}}$.
Projecting $\pi \cup \{\widehat e\}$ to the base graph yields a closed walk in $\Bk$ containing $e$ and avoiding $F_e$. Choose a closed subwalk of minimum positive length that contains $e$; this is a cycle, still contains $e$, and still avoids $F_e$. By Lemma~\ref{lem:cycle-hit}, this cycle has length greater than $2k$. Hence $\pi$ has length greater than $2k-1$, and so
$$|E(\pi)| = \dist_{H \setminus \widehat{F}_{\widehat{e}}}((u,i),(v,j)) > 2k-1.$$
On the other hand, $\widehat{e} \notin \widehat{F}_{\widehat{e}}$, so we have
$\dist_{\Bk^{\langle f \rangle} \setminus \widehat{F}_{\widehat{e}}}((u,i),(v,j)) = 1.$
This implies that $H$ is not an $f$-DFT $(2k-1)$-spanner of $\Bk^{\langle f \rangle}$. \qedhere

\end{proof}



\bibliography{refs}

\end{document}